\documentclass[12pt, a4paper]{article}%
\usepackage{amsmath,amssymb,eucal}
\usepackage{amsfonts}
\usepackage{mathrsfs}
\usepackage{slashed}%To produce Feynman slashed characters \slashed{\partial}
\numberwithin{equation}{section} 
\def\beq{\begin{equation}}
\def\eeq{\end{equation}}

\def\bC{ {{\mathbb{C}}}}
\def\bR{ {{\mathbb{R}}}}

\newcommand{\pk}[1]{p_{\kappa}}

\newtheorem{defn}{{\bf Definition}}[section]
\newtheorem{thm}[defn]{{\bf Theorem}}
\newtheorem{cor}[defn]{{\bf Corollary}}
\newtheorem{lem}[defn]{{\bf Lemma}}
\newtheorem{prop}[defn]{{\bf Proposition}}
\newtheorem{rem}[defn]{{\bf Remark}}

\newtheorem{notation}[defn]{Notation}
\newenvironment{proof}[1][Proof]{\textbf{#1.} }{\hfill \rule{0.5em}{0.5em}}
\begin{document}

\title{An unitary representation of inhomogeneous ${\rm SL}(2,\mathbb{C})$ using surfaces in $\mathbb{R}^4$}
\author{Adrian P. C. Lim \\
Email: ppcube@gmail.com
%ORCID: 0000-0002-2633-8734
}

\date{}

\maketitle

\begin{abstract}
We will construct a non-separable Hilbert space for which the inhomogeneous ${\rm SL}(2,\mathbb{C})$ acts on it unitarily. Each vector in this Hilbert space is described by a (rectangular) space-like surface in $\mathbb{R}^4$, for which a frame consisting of a time-like vector and a space-like vector, and a vector field is defined on it. The inner product on this Hilbert space is defined via a surface integral, which is associated with the area of the surface.
\end{abstract}

\hspace{.35cm}{\small {\bf MSC} 2020: }22E43 \\
\indent \hspace{.35cm}{\small {\bf Keywords}: ${\rm SL}(2,\bC)$, Lorentz transformation, unitary representation}

%{\bf \today}

%%%%%%%%%%%%%%%%%%%%%%%%%%%%%%%%%%%%%%%
%%%%%%%%%%%%%%%%%%%%%%%%%%%%%%%%%
%%%%%%%%%%%%%%%%%%%%%%%%%%%%%%%%%

%\tableofcontents

\section{Preliminaries}\label{s.pre}

Consider the 4-dimensional Euclidean space $\bR \times \bR^3\equiv \bR^4$. Note that $\bR$ will be referred to as the time-axis and $\bR^3$ is the spatial 3-dimensional Euclidean space. Fix the standard coordinates on $\bR^4\equiv \bR \times \bR^3$, $\vec{x} = (x^0, x^1, x^2, x^3)$, with time coordinate $x^0$ and spatial coordinates $x=(x^1, x^2, x^3)$. On $\bR^4$, we endow it with the standard Riemannian metric.

Throughout this article, we adopt Einstein's summation convention, i.e. we sum over repeated superscripts and subscripts.

We set the speed of light $c = 1$. On $\bR^4$, we can define the Minkowski metric, given by \beq \vec{x}\cdot \vec{y} := -x^0y^0 + \sum_{i=1}^3 x^iy^i. \label{e.inn.4} \eeq Note that our Minkowski metric is negative of the one used by physicists.

\section{Inhomogeneous ${\rm SL}(2,\bC)$}

A Lorentz transformation $\Lambda$ is a linear transformation mapping space-time $\bR^4$ onto itself, which preserves the Minkowski metric given in Equation (\ref{e.inn.4}). Indeed, the Lorentz transformations form a group, referred to as Lorentz group $L$. It has 4 components, and we will call the component containing the identity, as the restricted Lorentz group, denoted $L_+^\uparrow$.

Given any continuous group acting on $\bR^4$, we can consider its corresponding inhomogeneous group, whose elements are pairs consisting of a translation and a homogeneous transformation. For example, the Poincare group $\mathscr{P}$ containing the Lorentz group $L$, will have elements $\{\vec{a}, \Lambda\}$, where $\Lambda \in L$ and $\vec{a}$ will represent translation in the direction $\vec{a}$. The multiplication law for the Poincare group is given by \beq \{\vec{a}_1, A_1\}\{\vec{a}_2, A_2\} = \{\vec{a}_1 + A_1\vec{a}_2, A_1A_2\}. \nonumber \eeq

Associated with the restricted Lorentz group $L_+^\uparrow$ is the group of $2 \times 2$ complex matrices of determinant one, denoted by ${\rm SL}(2, \bC)$. There is an onto homomorphism $Y: {\rm SL}(2,\bC) \rightarrow L_+^\uparrow$. Thus, given $\Lambda \in {\rm SL}(2,\bC)$, $Y(\Lambda) \in L_+^\uparrow \subset L$. A formula for $Y$ can be found in \cite{streater}.

Instead of the Poincare group, we can consider the inhomogeneous ${\rm SL}(2,\bC)$ in its place, which we will denote also by ${\rm SL}(2, \bC)$, and use it to construct unitary representations of the Poincare group. Its elements will consist of $\{\vec{a}, \Lambda\}$ and its multiplication law is given by \beq \{\vec{a}_1, \Lambda_1\}\{\vec{a}_2, \Lambda_2 \} = \{\vec{a}_1+Y(\Lambda_1)\vec{a}_2, \Lambda_1 \Lambda_2\}. \nonumber \eeq By abuse of notation, for any $\Lambda \in {\rm SL}(2,\bC)$, we will write $\Lambda \vec{a}_2 \equiv Y(\Lambda)\vec{a}_2$.

Now, every irreducible finite dimensional representation of ${\rm SL}(2, \bC)$ is denoted by $D^{(j/2, k/2)}$, $j, k$ are non-negative integers. This representation is known as the spinor representation of ${\rm SL}(2, \bC)$. Any irreducible representation of ${\rm SU}(2)$ is equivalent to $A \in {\rm SU}(2)\subset {\rm SL}(2, \bC) \mapsto D^{(j/2,0)}(A)$ for some integer $j$. See \cite{streater}.

But what we are interested in is to construct unitary representations of ${\rm SL}(2,\bC)$. (See \cite{glimm1981quantum}.) And the only unitary finite dimensional representation of ${\rm SL}(2, \bC)$ is the trivial representation. See Theorem 16.2 in \cite{knapp2016representation}.

Thus, we need to look at unitary infinite dimensional representations of ${\rm SL}(2, \bC)$. A discussion of the classification of unitary representation of inhomogeneous ${\rm SL}(2, \bC)$ can be found in \cite{10.2307/1968551}. An explicit construction can be found in \cite{MR213473}. This construction used $L^2$ complex valued functions. Another construction is the space of tempered distributions (acting on $L^2(\bR^2)$), which can be found in \cite{ALBEVERIO197439}. An unitary representation of inhomogeneous ${\rm SL}(2, \bC)$ using a spinor representation of positive mass $m$, can be found in \cite{streater}.

We would like to give another construction, using rectangular surfaces in $\bR^4$. Clearly, the area of a surface is invariant under spatial rotation. Suppose we have a rectangular surface in the $x^0-x^1$ plane. If we give it a boost in the $x^1$-direction, because of time-dilation and length contraction, we see that the area of the rectangular surface remains invariant.

However, if we boost in the $x^2$ direction, then the area is no longer invariant. But this motivates us to look at surfaces in $\bR^4$, and we consider an inner product that is associated with the area of a surface. We can generalize this idea as given in the next section, by considering a field over a surface.

An unitary representation of inhomogeneous ${\rm SL}(2, \bC)$ appears in the Wightman's axioms. See \cite{streater}. A construction of a 4-dimensional quantum field theory that satisfies Wightman's axioms, is required to prove the Yang-Mills mass gap problem, as described in \cite{jaffe2006quantum}.

\section{Quantum Hilbert Space}

\begin{notation}
We will let $I = [0,1]$ be the unit interval, and $I^2 \equiv I \times I$. The variables $s, \bar{s}, t, \bar{t}$ will take values in $I$ and $\hat{s} = (s,\bar{s}),\ \hat{t} = (t,\bar{t})\in I^2$. And $d\hat{s} \equiv ds d\bar{s}$, $d\hat{t} \equiv dt d\bar{t}$. Typically, $s, \bar{s}, t, \bar{t}$ will be reserved as the variable for some parametrization, i.e. $\vec{\rho}: s \in I \mapsto \bR^4$.
\end{notation}

To construct a Hilbert space $\mathscr{H}$ for which ${\rm SL}(2,\bC)$ acts on unitarily, we need a finite dimensional real vector space $\mathcal{V}$, of which $\{E^\alpha\}_{\alpha=1}^N$ is its orthonormal basis. We endow a real inner product $\langle \cdot, \cdot \rangle$ on $\mathcal{V}$. Extend this inner product to be a sesquilinear complex inner product, over the complexification of $\mathcal{V}$, denoted as $\mathcal{V}_{\bC} \equiv \mathcal{V} \otimes_{\bR} \bC$. Hence, it is linear in the first variable, conjugate linear in the second.

\begin{defn}\label{d.ts.1}(Time-like and space-like surfaces)\\
Let $S$ be a rectangular surface in $\bR^4$, contained in some plane. By rotating the spatial axes if necessary, we may assume without any loss of generality, that a parametrization of $S$ is given by $\{(a^0 + sb^0, a^1 , a^2 + sb^2, a^3 +tb^3)^T \in \bR^4:\ s, t \in I\}$, for constants $a^\alpha, b^\alpha \in \bR$. Now, the surface $S$ is spanned by two directional vectors $(b^0, 0, b^2, 0)^T$ and $(0,0, 0, b^3)^T$. Note that $(b^0, 0, b^2, 0)^T$ lie in the $x^0-x^2$ plane and is orthogonal to $(0,0, 0, b^3)^T$.

We say a rectangular surface $S$ is space-like, if $|b^0| < |b^2|$, i.e. the acute angle which the vector $(b^0, 0, b^2, 0)^T$ makes with the $x^2$-axis in the $x^0-x^2$-plane is less than $\pi/4$.

We say a rectangular surface $S$ is time-like, if $|b^0| > |b^2|$, i.e. the acute angle which the vector $(b^0, 0, b^2, 0)^T$ makes with the $x^0$-axis in the $x^0-x^2$-plane is less than $\pi/4$.
\end{defn}

\begin{rem}\label{r.s.1}
%Let $e_1 = (1,0,0,0)$ be a directional vector in increasing time direction.
Let $S$ be a rectangular surface in $\bR^4$ contained in some plane, and $TS$ denote the set of directional vectors that lie inside $S$.

Write $\vec{v} = (v^0, v) \equiv (v^0, v^1, v^2, v^3) \in TS$, and define $|v|^2 = (v^1)^2 + (v^2)^2 + (v^3)^2$. An equivalent way to say that $S$ is time-like is \beq \inf_{\vec{0} \neq \vec{v} \in TS} \frac{|v|^2}{(v^0)^2} < 1. \label{e.inf.1} \eeq

And we say that $S$ is space-like if \beq \inf_{\vec{0} \neq \vec{v} \in TS} \frac{|v|^2}{(v^0)^2} > 1. \label{e.inf.2} \eeq
\end{rem}

By definition, a time-like surface must contain a time-like directional vector in it. Since under Lorentz transformation, a time-like vector remains time-like, we see that a time-like surface remains time-like under Lorentz transformation. Similarly, a surface is space-like means all its directional vectors in the surface are space-like. Under Lorentz transformation, all its directional vectors spanning $S$ remain space-like, hence a space-like surface remains space-like under Lorentz transformation.

In the rest of this article, we will only consider surfaces which contain rectangular surfaces contained in some plane, which are space-like.

Suppose we consider a time-like surface contained in some plane, i.e. let $\phi \in \bR$ such that $\coth\phi = |b^0|/|b^2| > 1$. Then we can write $b^0 = \alpha\cosh \phi$, $b^2 = \alpha\sinh \phi$, $\phi \geq 0$. When the surface is space-like, then we will write $b^0 = \alpha\sinh \phi$, $b^2 = \alpha\cosh \phi$, $\alpha$ is a non-zero constant. Hence, we can parametrize a time-like rectangular surface with \beq (s,t) \in I^2 \mapsto
\left(
  \begin{array}{c}
    a^0 + s\alpha\cosh \phi \\
    a^1  \\
    a^2 + s\alpha\sinh \phi \\
    a^3 + t\beta b^3\\
  \end{array}
\right) =
\left(
  \begin{array}{c}
    a^0 \\
    a^1 \\
    a^2 \\
    a^3 \\
  \end{array}
\right) +
\left(
  \begin{array}{cccc}
    \cosh\phi &\ 0 &\ \sinh \phi &\ 0 \\
    0 &\ 1 &\ 0 &\ 0 \\
    \sinh\phi &\ 0 &\ \cosh\phi &\ 0 \\
    0 &\ 0 &\ 0 &\ 1 \\
  \end{array}
\right)
\left(
  \begin{array}{c}
    \alpha s \\
    0 \\
    0 \\
    \beta t b^3\\
  \end{array}
\right); \nonumber \eeq
a space-like rectangular surface with \beq (s,t) \in I^2 \mapsto
\left(
  \begin{array}{c}
    a^0 + s\alpha\sinh \phi \\
    a^1  \\
    a^2 + s\alpha\cosh \phi \\
    a^3 + t\beta b^3\\
  \end{array}
\right) =
\left(
  \begin{array}{c}
    a^0 \\
    a^1 \\
    a^2 \\
    a^3 \\
  \end{array}
\right) +
\left(
  \begin{array}{cccc}
    \cosh\phi &\ 0 &\ \sinh\phi &\ 0 \\
    0 &\ 1 &\ 0 &\ 0 \\
    \sinh\phi &\ 0 &\ \cosh\phi &\ 0 \\
    0 &\ 0 &\ 0 &\ 1 \\
  \end{array}
\right)
\left(
  \begin{array}{c}
    0 \\
    0 \\
     \alpha s \\
    \beta tb^3 \\
  \end{array}
\right), \label{e.a.1} \eeq whereby $\alpha$ and $\beta$ are some fixed non-zero constants.

In this article, when we say surface $S$, we mean it is a finite, disjoint union of space-like rectangular compact surfaces in $\bR^4$, containing none, some or all of its boundary points.

\begin{defn}(Surface)\label{d.sur.1}\\
Any surface $S \equiv \{S_u: u=1,\ldots, n\} \subset \bR^4$, satisfies the following conditions:
\begin{itemize}
  \item each connected component $S_u$ is a rectangular surface in $\bR^4$ contained in some plane, which is space-like;
  \item each connected component $S_u$ may contain none, some or all of its boundary;
  \item $S_u \cap S_v = \emptyset$ if $u \neq v$;
  \item $S$ is contained in some bounded set in $\bR^4$.
\end{itemize}
\end{defn}

Given 2 surfaces, $S$ and $\tilde{S}$, we need to take the intersection and union of these surfaces. Now, the union of these 2 surfaces can always be written as a disjoint union of connected sets, each such set is a surface, containing none, some or all of its boundary points. However, the intersection may not be a surface. For example, the two surfaces may intersect to give a curve. In such a case, we will take the intersection to be the empty set $\emptyset$.

\begin{defn}\label{d.a.1}
Let $S_0$ be a compact rectangular space-like surface inside the $x^2-x^3$ plane. From Equation (\ref{e.a.1}), we see that any rectangular space-like surface $S$ can be transformed to $S_0$ by Lorentz transformations and translation. Let $e_0 = (1, 0, 0, 0)^T$ and $e_1 = (0, 1, 0, 0)^T$ be time-like and space-like vectors respectively.

We say that $\{\hat{f}_0, \hat{f}_1\}$ is a frame for a compact space-like surface $S$ contained in some plane, if there exists a sequence of Lorentz transformations $\Lambda_1, \cdots, \Lambda_n$ and a translation by $\vec{a} \in \bR^4$, such that
\begin{itemize}
  \item $S = \Lambda_n \cdots \Lambda_1 S_0 + \vec{a}$;
  \item $\hat{f}_0 = \Lambda_n \cdots \Lambda_1 e_0 \in \bR^4$; and
  \item $\hat{f}_1 = \Lambda_n \cdots \Lambda_1 e_1 \in \bR^4$.
\end{itemize}
\end{defn}

\begin{rem}
Observe that $\hat{f}_0$ is time-like and $\hat{f}_1$ is space-like, both in $\mathbb{R}^4$, satisfying the following properties:
\begin{itemize}
  \item $\hat{f}_0 \cdot \vec{x} = \hat{f}_1 \cdot \vec{x} = 0$ for any directional vector $\vec{x}$ in $S$;
  \item $\hat{f}_0 \cdot \hat{f}_1 = 0$ and
  \item $\hat{f}_0\cdot \hat{f}_0 = -1$ and $\hat{f}_1 \cdot \hat{f}_1 = 1$.
\end{itemize}
\end{rem}

Given a space-like surface $S = \bigcup_{u=1}^\infty S_u$, each $S_u$ is a compact rectangular space-like surface contained inside some plane, we will write $\{\hat{f}_0^u, \hat{f}_1^u\}_{u\geq 1}$ to denote a collection of frames for $S$, such that each $\{\hat{f}_0^u, \hat{f}_1^u\}$ is a frame for $S_u$.

Consider the trivial bundle $\bR^4 \times \mathcal{V}_{\bC} \rightarrow \bR^4$. The Hilbert space $\mathscr{H}$ will consist of vectors of the form $\sum_{u=1}^\infty(S_u, f_\alpha^u \otimes E^\alpha, \{\hat{f}_0^u, \hat{f}_1^u\})$, whereby $S_u$ is some space-like surface in $\bR^4$ as described above, $f_\alpha^u$ will be some (measurable) complex-valued function, which is defined on the surface $S_u$ and $\{\hat{f}_0^u, \hat{f}_1^u\}$ is a frame for each surface $S_u$ contained in a plane as described in Definition \ref{d.a.1}. We sum over repeated index $\alpha$, from $\alpha=1$ to $N$. Let $\sigma: [0,1] \times [0,1] \rightarrow \bR^4$ be a parametrization of $S$. The complex-valued function $f_\alpha$ is said to be measurable on $S$, if $f_\alpha \circ \sigma: [0,1]^2 \rightarrow \bC$ is measurable.

\begin{rem}
One should think of $f_\alpha^u \otimes E^\alpha$ as a section of the vector bundle $\bR^4 \times \mathcal{V}_{\bC} \rightarrow \bR^4$, defined over the surface $S_u$. Henceforth, it will be referred to as a field over $S_u$. If $\vec{x} \in S_u$, then the field vector at $\vec{x}$ will be given by $f_\alpha^u(\vec{x}) \otimes E^\alpha$.
\end{rem}

Let $S$ and $\tilde{S}$ be rectangular space-like surfaces contained in a plane. Given scalars $\lambda$ and $\mu$, we define the addition and scalar multiplication as
\begin{align}
\lambda\left(S, f_\alpha \otimes E^\alpha, \{\hat{f}_0, \hat{f}_1\}\right)& + \mu\left(\tilde{S}, g_\alpha \otimes E^\alpha, \{\hat{f}_0, \hat{f}_1\}\right) \nonumber \\
&:= \left(S \cup \tilde{S}, (\lambda\tilde{f}_\alpha + \mu\tilde{g}_\alpha) \otimes E^\alpha, \{\hat{f}_0, \hat{f}_1\}\right). \label{e.v.1}
\end{align}
Here, we extend $f_\alpha$ to be $\tilde{f}_\alpha: S \cup \tilde{S} \rightarrow \bR$ by $\tilde{f}_\alpha(p) = f_\alpha(p)$ if $p \in S$; otherwise $\tilde{f}(p) = 0$. Similarly, $\tilde{g}_\alpha$ is an extension of $g_\alpha$, defined as $\tilde{g}_\alpha(p) = g_\alpha(p)$, if $p \in \tilde{S}$, 0 otherwise.

\begin{rem}
For the above addition to hold, we require that the frame $\{\hat{f}_0, \hat{f}_1\}$ on $S$ and $\tilde{S}$ to be the same.
\end{rem}

\begin{defn}
Given a bounded surface $S = \bigcup_{u=1}^n S_u$, each $S_u$ equipped with a frame $\{\hat{f}_0^u, \hat{f}_1^u\}$, and a set of bounded and continuous complex-valued functions $\{f_\alpha^u\}_{\alpha=1}^N$, $u \geq 1$, defined on $S_u$, form a vector $\sum_{u=1}^n(S_u, f_\alpha^u \otimes E^\alpha, \{\hat{f}_0^u, \hat{f}_1^u\})$. Note that for each $u$, $f_\alpha^u \otimes E^\alpha$ is a vector field over $S_u$, with $S_u$ contained in some space-like plane, equipped with a frame $\{\hat{f}_0^u, \hat{f}_1^u\}$. Let $V$ be a (complex) vector space containing such vectors, with addition and scalar multiplication defined by Equation (\ref{e.v.1}).
\end{defn}

\begin{rem}
The zero vector can be written as $(S, 0, \{\hat{f}_0, \hat{f}_1\})$ for any space-like rectangular surface $S$ contained in a plane, equipped with any frame $\{\hat{f}_0, \hat{f}_1\}$.
\end{rem}

We want to make $V$ to be a normed space. We can define the following inner product.

\begin{defn}\label{d.inn.1}
For a surface $S$, define  $\int_S d\rho$ as in Corollary \ref{c.w.4}. Assume that $S$ and $\tilde{S}$ be space-like surfaces, contained in a plane. Let $\sigma: I^2 \rightarrow S \cap \tilde{S}$ be a parametrization. Define an inner product $\langle \cdot, \cdot \rangle$ for $\left(S, f_\alpha \otimes E^\alpha, \{\hat{f}_0, \hat{f}_1\} \right) \in V$, \\ $\left(\tilde{S}, g_\beta \otimes E^\beta, \{\hat{g}_0, \hat{g}_1\}\right) \in V$, given by
\begin{align*}
&\left\langle \left(S, f_\alpha \otimes E^\alpha, \{\hat{f}_0, \hat{f}_1\}\right), \left(\tilde{S}, g_\beta \otimes E^\beta, \{\hat{g}_0, \hat{g}_1\}\right) \right\rangle \\
&\hspace{1cm} :=  \sum_{\alpha,\beta=1}^N \int_{S \cap \tilde{S}}  f_{\alpha}\overline{g_{\beta}}\
d\rho\ \langle E^\alpha , E^\beta \rangle \\
&\hspace{1cm} = \sum_{\alpha,\beta=1}^N \sum_{0 \leq a<b \leq 3}\int_{I^2}  f_{\alpha}(\sigma(\hat{s}))\overline{g_{\beta}}(\sigma(\hat{s}))
\rho_\sigma^{ab}(\hat{s})|J_{ab}^\sigma|(\hat{s}) d\hat{s}\ \langle E^\alpha , E^\beta \rangle,
\end{align*}
provided $\{\hat{f}_0, \hat{f}_1\} = \{\hat{g}_0, \hat{g}_1\}$. Otherwise, it is defined as zero.
\end{defn}

For a surface $S$, $\int_S d\rho$ gives us the area of a surface $S$. See Corollary \ref{c.w.4} and its following remark. Unfortunately, area of a surface is not invariant under boost. Therefore, the above inner product will not be invariant if we boost $S$. Hence we will not consider this inner product.

Given a surface $S$, let $\sigma$ be any parametrization of $S$. We can define $\int_S d\rho$ using this parametrization $\sigma$ as given in Definition \ref{d.r.1}. Now, replace $\sigma \equiv (\sigma_0, \sigma_1, \sigma_2, \sigma_3)$ with $\acute{\sigma} = (i\sigma_0, \sigma_1, \sigma_2, \sigma_3)$ and hence define $\int_S d\acute{\rho}$ as given in Definition \ref{d.r.2}. Essentially, we change the time component in the formula for $\int_S d\rho$ to be purely imaginary. Now we will define and use the following inner product in the rest of this article.

\begin{defn}\label{d.inn.2}
Refer to Definition \ref{d.r.2} for the definition of $\int_S d|\acute{\rho}|$, for a surface $S$. Assume that $S$ and $\tilde{S}$ be space-like surfaces, contained in a plane. Let $\sigma: I^2 \rightarrow S \cap \tilde{S}$ be a parametrization. Define an inner product $\langle \cdot, \cdot \rangle$ for $\left(S, f_\alpha \otimes E^\alpha, \{\hat{f}_0, \hat{f}_1\} \right) \in V$, $\left(\tilde{S}, g_\beta \otimes E^\beta, \{\hat{g}_0, \hat{g}_1\}\right) \in V$, given by
\begin{align*}
&\left\langle \left(S, f_\alpha \otimes E^\alpha, \{\hat{f}_0, \hat{f}_1\}\right), \left(\tilde{S}, g_\beta \otimes E^\beta, \{\hat{g}_0, \hat{g}_1\}\right) \right\rangle :=  \sum_{\alpha,\beta=1}^N \int_{S \cap \tilde{S}}  f_{\alpha}\overline{g_{\beta}}\ d|
\acute{\rho}|\ \langle E^\alpha , E^\beta \rangle \\
&= \sum_{\alpha,\beta=1}^N \int_{I^2}  f_{\alpha}(\sigma(\hat{s}))\overline{g_{\beta}}(\sigma(\hat{s}))
\left|\sum_{0 \leq a<b \leq 3}\acute{\rho}_\sigma^{ab}(\hat{s})[\det \acute{J}_{ab}^\sigma](\hat{s})\right|d\hat{s}\ \langle E^\alpha , E^\beta \rangle,
\end{align*}
provided $\{\hat{f}_0, \hat{f}_1\} = \{\hat{g}_0, \hat{g}_1\}$. Otherwise, it is zero.

Denote its norm by $|\cdot|$. Let $\mathscr{H}$ denote the Hilbert space containing $V$, using the said inner product.
\end{defn}

\begin{rem}
The integrals in Definitions \ref{d.inn.1} and \ref{d.inn.2} are independent of the choice of parametrization $\sigma$.
\end{rem}

\begin{prop}
The Hilbert space $\mathscr{H}$ is non-separable.
\end{prop}

\begin{proof}
Consider a compact rectangular surface $S_0$ contained in the $x^2-x^3$ plane, with $\{e_0, e_1\}$ as its frame. Then, we see that \beq \left\{\left(S_0 + \vec{a}, E^\alpha, \{e_0, e_1\}\right):\ \vec{a} \in \mathbb{R} \right\} \nonumber \eeq is an uncountable set of vectors in $\mathscr{H}$, since \beq \left\langle \Big(S_0 + \vec{a}, E^\alpha, \{e_0, e_1\}\Big), \left(S_0 + \vec{b}, E^\alpha, \{e_0, e_1\}\right) \right\rangle = 0, \nonumber \eeq if $\vec{a} \neq \vec{b}$. Thus this uncountable set consists of an orthogonal vectors, hence the Hilbert space is non-separable.
\end{proof}

\section{Unitary representation of inhomogeneous ${\rm SL}(2,\bC)$}

Given a vector $\vec{x} \in \bR^4$, $\{\vec{a}, \Lambda\}$ acts on $\vec{x}$ by $\vec{x} \mapsto \Lambda\vec{x} + \vec{a}$. By abuse of notation, for a surface $S$, $\{\vec{a}, \Lambda\}$ acts on $S$ by $S \mapsto \Lambda S + \vec{a}$, which means first apply a Lorentz transformation $Y(\Lambda)$ to the surface $S$, followed by translating the surface $Y(\Lambda)S$ by $\vec{a}$.

Let $\langle \vec{x}, \vec{y} \rangle := \sum_{a=0}^3 x^a y^a$ be the usual scalar product in $\bR^4$. (Compare with the Minkowski metric given in Equation (\ref{e.inn.4}).) We will now define a unitary representation of the inhomogeneous ${\rm SL}(2,\bC)$, acting on $\mathscr{H}$.

\begin{defn}(Unitary Representation of the inhomogeneous ${\rm SL}(2,\bC)$)\\
Choose 2 numbers $\hat{H}$ and $\hat{P}$, which are fixed. There is a continuous unitary representation of the inhomogeneous ${\rm SL}(2,\bC)$, $\{\vec{a}, \Lambda\} \mapsto U(\vec{a}, \Lambda)$. Now, $U(\vec{a},\Lambda)$ acts on the Hilbert space $\mathscr{H}$, by \beq \left(S, f_\alpha \otimes E^\alpha, \{\hat{f}_0, \hat{f}_1\}\right) \mapsto U(\vec{a},\Lambda)\left(S, f_\alpha \otimes E^\alpha, \{\hat{f}_0, \hat{f}_1\}\right), \nonumber \eeq as
\begin{align}
U(\vec{a},\Lambda)&\left(S, f_\alpha \otimes E^\alpha, \{\hat{f}_0, \hat{f}_1\}\right) \nonumber \\
&:= \left(\Lambda S + \vec{a}, e^{-i[\vec{a}\cdot ( \hat{H}\Lambda\hat{f}_0 + \hat{P}\Lambda\hat{f}_1)] } f_\alpha[\Lambda^{-1}(\cdot- \vec{a})] \otimes E^{\alpha} , \{\Lambda \hat{f}_0, \Lambda \hat{f}_1\}\right). \label{e.u.4}
\end{align}
Here, $S$ is a space-like surface contained in some plane.
\end{defn}

\begin{rem}
\begin{enumerate}
  \item The unitary representation depends on $\hat{H}, \hat{P} \in \bR$, which are fixed.
  \item Let us explain the formula on the RHS of Equation(\ref{e.u.4}). Suppose $\sigma: I^2 \rightarrow S$ is a parametrization for $S$. Then $\Lambda \sigma + \vec{a} \equiv Y(\Lambda)\sigma + \vec{a}$ will be a parametrization for $Y(\Lambda)S+ \vec{a}$. And, the field at the point $\vec{x} := Y(\Lambda)\sigma(\hat{s}) + \vec{a} \in \Lambda S + \vec{a}$, is given by
      \begin{align*}
      e^{-i[\vec{a}\cdot (\hat{H}Y(\Lambda)\hat{f}_0 + \hat{P}Y(\Lambda)\hat{f}_1)]}& f_\alpha[Y(\Lambda^{-1})(\vec{x}- \vec{a})] \otimes E^{\alpha} \\
      &\equiv e^{-i[\vec{a}\cdot (\hat{H}Y(\Lambda)\hat{f}_0 + \hat{P}Y(\Lambda)\hat{f}_1)]} f_\alpha[\sigma(\hat{s})]\otimes E^\alpha.
      \end{align*}
\end{enumerate}
\end{rem}

We will now prove that $U(\vec{a},\Lambda)$ is unitary and that $U$ is a unitary representation of the inhomogeneous group ${\rm SL}(2, \bC)$.

\begin{thm}\label{t.b.2}
The map $U(\vec{a}, \Lambda)$ defined on $\mathscr{H}$ is unitary, using the inner product $\langle \cdot, \cdot \rangle$ as defined in Definition \ref{d.inn.2}.
\end{thm}

\begin{proof}
We will first show that it is a representation. Given $\{\vec{b}, \Gamma\}, \{\vec{a}, \Lambda\}$ from ${\rm SL}(2, \bC)$, we have ($\hat{g}_0 = \Lambda \hat{f}_0, \hat{g}_1 = \Lambda \hat{f}_1$)
\begin{align*}
&U(\vec{b},\Gamma)U(\vec{a},\Lambda)\left(S, f_\alpha \otimes E^\alpha, \{\hat{f}_0, \hat{f}_1\}\right) \\
&= U(\vec{b},\Gamma)\left(\Lambda S + \vec{a}, e^{-i[\vec{a}\cdot (\hat{H}\hat{g}_0 + \hat{P}\hat{g}_1)] } f_\alpha[\Lambda^{-1}(\cdot- \vec{a})] \otimes E^{\alpha} , \{\hat{g}_0,  \hat{g}_1\}\right) \\
&= \left(\Gamma\Lambda S + \Gamma\vec{a} + \vec{b},
\hat{E} f_\alpha[\Lambda^{-1}\Gamma^{-1}(\cdot- \vec{b}) - \Lambda^{-1}\vec{a})] \otimes E^{\alpha} , \{\Gamma \hat{g}_0, \Gamma \hat{g}_1\}\right) \\
&= \left(\Gamma\Lambda S + \Gamma\vec{a} + \vec{b},
\hat{E} f_\alpha[(\Gamma\Lambda)^{-1}(\cdot- \vec{b} - \Gamma\vec{a})] \otimes E^{\alpha} , \{\Gamma\Lambda \hat{f}_0, \Gamma\Lambda \hat{f}_1\}\right),
\end{align*}
whereby $\hat{E}$ is equal to
\begin{align*}
&e^{-i[\vec{b}\cdot (\hat{H}\Gamma\hat{g}_0 + \hat{P}\Gamma\hat{g}_1)] }e^{-i[\vec{a}\cdot ( \hat{H}\hat{g}_0 + \hat{P}\hat{g}_1)] }
= e^{-i[(\vec{b} + \Gamma \vec{a})\cdot \Gamma\Lambda( \hat{H}\hat{f}_0 + \hat{P}\hat{f}_1) ]},
\end{align*}
because $\Gamma \vec{x} \cdot \Gamma \vec{y} = \vec{x} \cdot \vec{y}$.

Hence, we see that
\begin{align*}
U(\vec{b},\Gamma)U(\vec{a},\Lambda)\left(S, f_\alpha \otimes E^\alpha, \{\hat{f}_0, \hat{f}_1\}\right)
= U(\vec{b} + \Gamma\vec{a},\Gamma\Lambda)\left(S, f_\alpha \otimes E^\alpha, \{\hat{f}_0, \hat{f}_1\}\right).
\end{align*}

Write \beq \acute{\rho}_\sigma = \sum_{0\leq a < b \leq 3}\acute{\rho}_\sigma^{ab} \det \acute{J}_{ab}^\sigma. \nonumber \eeq Note that multiplication by $e^{ic}$ is unitary, $c \in \bR$. By definition, $\acute{\rho}_\sigma$ is invariant under translation. From Equation (\ref{e.u.4}), it is clear that the inner product in Definition \ref{d.inn.2} will be invariant under translation $U(\vec{a},1)$. We will now check that it is unitary under a Lorentz transformation $\Lambda$.

Let $\sigma: I^2 \rightarrow \bR^4$ be a parametrization of a compact rectangular surface $S$. Under a Lorentz transformation $\Lambda$, the surface $\Lambda S$ can be parametrized by $\Lambda \sigma$. From Equation (\ref{e.j.6}), we have that $\acute{\rho}_{\Lambda \sigma} = \acute{\rho}_{\sigma}$, pointwise under Lorentz transformation. From Definition \ref{d.inn.2},
\begin{align}
\sum_{\alpha,\beta=1}^N &\int_{I^2}  f_{\alpha}(\sigma(\hat{s}))\overline{g_{\beta}}(\sigma(\hat{s}))
|\acute{\rho}_\sigma|(\hat{s}) d\hat{s}\ \langle E^\alpha , E^\beta \rangle \nonumber \\
&\longmapsto_{U(0, \Lambda)} \sum_{\alpha,\beta=1}^N \int_{I^2}  f_{\alpha}(\sigma(\hat{s}))\overline{g_{\beta}}(\sigma(\hat{s}))
|\acute{\rho}_{\Lambda \sigma}|(\hat{s}) d\hat{s}\ \langle E^\alpha , E^\beta \rangle \nonumber \\
&= \sum_{\alpha,\beta=1}^N \int_{I^2}  f_{\alpha}(\sigma(\hat{s}))\overline{g_{\beta}}(\sigma(\hat{s}))
|\acute{\rho}_\sigma|(\hat{s}) d\hat{s}\ \langle E^\alpha , E^\beta \rangle. \label{e.inn.10}
\end{align}

\end{proof}

\appendix

\section{Lorentz transformation}\label{a.lzt}

For the reader who is not familiar with Lorentz transformation, we have included this section for his convenience. Let $c>0$ be the speed of light. Given two frames $F$ and $F'$, let $F'$ move with velocity $v$ in the $x^1$ direction. Let $\beta = v/c$ and because $v < c$, we have $-1<\beta < 1$.

Suppose $(x^0, x)$ are the coordinates of an event recorded by an observer in $F$ and  $(\tilde{x}^0, \tilde{x})$ are the coordinates of the same event recorded by an observer in $F'$. Then, we have the following coordinate transformation $\{x^a\}_{a=0}^3 \mapsto \{\tilde{x}^a\}_{a=0}^3$,
\beq
\left(
  \begin{array}{c}
    \tilde{x}^0 \\
    \tilde{x}^1 \\
  \end{array}
\right) =
\gamma
\left(
  \begin{array}{cc}
    1 &\ -\frac{\beta}{c} \\
    -\beta c &\ 1 \\
  \end{array}
\right)
\left(
  \begin{array}{c}
    x^0 \\
    x^1 \\
  \end{array}
\right),\qquad \tilde{x}^2 = x^2,\quad \tilde{x}^3 = x^3, \label{e.lz.1}
\eeq whereby \beq \gamma = \frac{1}{\sqrt{1 - \beta^2}}. \nonumber \eeq Such a transformation is linear and we say it is a Lorentz boost in the $x^1$ direction. One can write down similar transformations for boosts in the $x^2$ and $x^3$ directions.

Its inverse transformation will be given by
\beq
\left(
  \begin{array}{c}
    x^0 \\
    x^1 \\
  \end{array}
\right) =
\gamma
\left(
  \begin{array}{cc}
    1 &\ \frac{\beta}{c} \\
    \beta c &\ 1 \\
  \end{array}
\right)
\left(
  \begin{array}{c}
    \tilde{x}^0 \\
    \tilde{x}^1 \\
  \end{array}
\right),\qquad x^2 = \tilde{x}^2,\quad x^3 = \tilde{x}^3. \nonumber
\eeq

We are going to write the transformation in terms of hyperbolic functions. Define $\zeta$ via $\beta = \tanh \zeta$ or $\zeta = \tanh^{-1}\beta$ which is well-defined, because $-1< \beta < 1$. Then, $1 - \beta^2 = {\rm sech}^2 \zeta$, so we have $\gamma= \cosh \zeta$. Hence, $\beta \gamma = \sinh \zeta$.

We can now rewrite Equation (\ref{e.lz.1}) in terms of hyperbolic functions, \beq
\left(
  \begin{array}{c}
    \tilde{x}^0 \\
    \tilde{x}^1 \\
  \end{array}
\right) =
\left(
  \begin{array}{cc}
    \cosh \zeta &\ -\frac{\sinh \zeta}{c} \\
    -c\sinh \zeta &\ \cosh\zeta \\
  \end{array}
\right)
\left(
  \begin{array}{c}
    x^0 \\
    x^1 \\
  \end{array}
\right),\quad \tilde{x}^2 = x^2,\quad \tilde{x}^3 = x^3. \nonumber
\eeq

Compare this with rotation about the $x^3$-axis, \beq
\left(
  \begin{array}{c}
    x^1 \\
    x^2 \\
  \end{array}
\right)\ \longmapsto
\left(
  \begin{array}{cc}
    \cos \theta &\ -\sin \theta \\
    \sin \theta &\ \cos \theta \\
  \end{array}
\right)
\left(
  \begin{array}{c}
    x^1 \\
    x^2 \\
  \end{array}
\right),\quad x^3 \mapsto x^3. \nonumber \eeq

\section{Surface Integrals}

Let $S$ be a surface embedded in $\bR^4$ and $\sigma\equiv ( \sigma_0, \sigma_1, \sigma_2, \sigma_3): [0,1]^2 \equiv I^2 \rightarrow \bR^4$ be its parametrization. Here, $\sigma' = \partial \sigma/\partial s$ and $\dot{\sigma} = \partial \sigma/\partial t$.

\begin{defn}\label{d.r.1}
Let $\sigma: [0,1]^2 \equiv I^2 \rightarrow \bR^4$ be a parametrization of a surface $S \subset \bR^4$.
For $a,b=0,1,2,3$, define Jacobian matrices,
\begin{align}
J_{ab}^\sigma(s,t) = \left(
               \begin{array}{cc}
                 \sigma_a'(s,t) & \dot{\sigma}_a(s,t) \\
                 \sigma_b'(s,t) & \dot{\sigma}_b(s,t) \\
               \end{array}
             \right),\ a \neq b, \nonumber
\end{align}
and write $|J^\sigma_{ab}| = \sqrt{[\det{J^\sigma_{ab}}]^2}$ and $W_{ab}^{ cd} := J_{cd}^\sigma J_{ab}^{\sigma, -1}$, $a, b, c, d$ all distinct. Note that $W_{cd}^{ab} = (W_{ab}^{cd})^{-1}$.

For $a,b, c, d$ all distinct, define $\rho_\sigma^{ab}: I^2 \rightarrow \bR$ by
\begin{align}
\rho_\sigma^{ab} =& \frac{1}{\sqrt{\det\left[ 1+ W_{ab}^{cd,T}W_{ab}^{cd}\right]}} \equiv \frac{|J_{ab}^\sigma|}{\sqrt{\det\left[J_{ab}^{\sigma,T}J_{ab}^\sigma + J_{cd}^{\sigma,T}J_{cd}^\sigma\right]}}. \label{e.ri.1}
\end{align}
\end{defn}

\begin{cor}\label{c.w.4}
Define
\begin{align}
\int_S d\rho := \sum_{0 \leq a<b \leq 3}\int_{I^2}\rho_\sigma^{ab}(s,t)|J_{ab}^\sigma|(s,t)
\ ds dt, \nonumber
\end{align}
which gives us the area of the surface $S$.
\end{cor}

\begin{rem}
Note that $\int_S d\rho$ is clearly independent of the choice of parametrization used. When $S$ is a rectangular surface in $x^i-x^j$ plane, a direct calculation will show that it gives us the area of a surface $S$. To show that it is indeed the area, we will show in the next lemma, that it is independent of the choice of orthonormal basis, hence showing that it is the area.
\end{rem}

\begin{lem}\label{l.i.1}
We have $\int_S d\rho$ is independent of the orthonormal basis used in $\bR^4$.
\end{lem}

\begin{proof}
We let $\langle \cdot, \cdot \rangle$ denote the standard inner product on $\bR^4$, and $\{e_0, e_1, e_2, e_3\}$ be an orthonormal basis for $\bR^4$. Note that $\langle \cdot, \cdot \rangle$ will induce an inner product $\langle \cdot, \cdot \rangle_2$ on the second exterior power $\Lambda^2 \bR^4$, and \beq \{e_0\wedge e_1, e_0 \wedge e_2, e_0 \wedge e_3, e_1 \wedge e_2, e_3 \wedge e_1, e_1 \wedge e_2 \} \nonumber \eeq is an orthonormal basis. See \cite{MR1312606}.

Let $\hat{\sigma} = A\sigma$, $A$ is an orthogonal matrix. We will first show that $J_{ab}^{\sigma,T}J_{ab}^\sigma + J_{cd}^{\sigma,T}J_{cd}^\sigma$ is independent of any orthonormal basis used. We will only show for $a=0$, $b=1$, $c = 2$, $d = 3$.

Now, we see that \beq
\left(
  \begin{array}{c}
    J_{01}^{\hat{\sigma}} \\
    J_{23}^{\hat{\sigma}} \\
  \end{array}
\right)
 = A(\sigma', \dot{\sigma})\equiv (A\sigma', A\dot{\sigma}), \nonumber \eeq therefore
\begin{align}
J_{01}^{\hat{\sigma},T}J_{01}^{\hat{\sigma}} + J_{23}^{\hat{\sigma},T}J_{23}^{\hat{\sigma}} = &
\left(
  \begin{array}{c}
    \sigma^{\prime, T} \\
    \dot{\sigma}^T \\
  \end{array}
\right)
A^T A(\sigma', \dot{\sigma})= \left(
  \begin{array}{c}
    \sigma^{\prime, T} \\
    \dot{\sigma}^T \\
  \end{array}
\right)(\sigma', \dot{\sigma}) =
\left(
  \begin{array}{cc}
    \langle \sigma', \sigma' \rangle &\ \langle \sigma', \dot{\sigma} \rangle \\
    \langle \dot{\sigma}, \sigma' \rangle &\ \langle \dot{\sigma}, \dot{\sigma} \rangle \\
  \end{array}
\right).\label{e.j.2}
\end{align}
Indeed, we see that \beq J_{ab}^{\hat{\sigma},T}J_{ab}^{\hat{\sigma}} + J_{cd}^{\hat{\sigma},T}J_{cd}^{\hat{\sigma}} = \left(
  \begin{array}{cc}
    \langle \sigma', \sigma' \rangle &\ \langle \sigma', \dot{\sigma} \rangle \\
    \langle \dot{\sigma}, \sigma' \rangle &\ \langle \dot{\sigma}, \dot{\sigma} \rangle \\
  \end{array}
\right), \label{e.j.1} \eeq for any distinct $a,b,c,d$, and hence \beq \det\left[J_{ab}^{\hat{\sigma},T}J_{ab}^{\hat{\sigma}} + J_{cd}^{\hat{\sigma},T}J_{cd}^{\hat{\sigma}} \right] = \langle \sigma'\wedge \dot{\sigma}, \sigma'\wedge \dot{\sigma} \rangle_2. \nonumber \eeq

Now \beq \det J_{ab}^\sigma = \det\left(
\begin{array}{cc}
  \langle \sigma', e_a \rangle &\ \langle \dot{\sigma}, e_a \rangle \\
  \langle \sigma', e_b \rangle &\ \langle \dot{\sigma}, e_b \rangle
\end{array}
\right) = \langle \sigma' \wedge \dot{\sigma}, e_a \wedge e_b \rangle_2. \nonumber \eeq
Therefore, \beq \det J_{ab}^{\hat{\sigma}} = \det J_{ab}^{A\sigma} = \langle A\sigma' \wedge A\dot{\sigma}, e_a \wedge e_b \rangle_2 = \langle \sigma' \wedge \dot{\sigma}, A^Te_a \wedge A^Te_b \rangle_2. \nonumber \eeq Note that the linear transformation $e_a \wedge e_b \mapsto A^T e_a \wedge A^T e_b$ is an orthogonal transformation, if $A$ is an orthogonal $4 \times 4$ matrix.

Hence,
\begin{align}
\sum_{0 \leq a<b \leq 3}\rho_\sigma^{ab}|J_{ab}^\sigma| =& \sum_{0 \leq a<b \leq 3}\frac{|J_{ab}^\sigma|^2}{\sqrt{\det\left[J_{ab}^{\sigma,T}J_{ab}^\sigma + J_{cd}^{\sigma,T}J_{cd}^\sigma\right]}} \nonumber \\
\longmapsto& \sum_{0 \leq a<b \leq 3}\frac{|J_{ab}^{A\sigma}|^2}{\sqrt{\det \left[J_{ab}^{A\sigma,T}J_{ab}^{A\sigma} + J_{cd}^{A\sigma,T}J_{cd}^{A\sigma}\right]}} \nonumber \\
=& \sum_{0 \leq a<b \leq 3}\frac{|J_{ab}^{A\sigma}|^2}{\sqrt{\det\left[J_{ab}^{\sigma,T}J_{ab}^{\sigma} + J_{cd}^{\sigma,T}J_{cd}^{\sigma}\right]}} \nonumber \\
=& \sum_{0 \leq a<b \leq 3}\left[ \frac{\langle \sigma' \wedge \dot{\sigma}, A^Te_a \wedge A^Te_b \rangle_2}{\langle \sigma'\wedge \dot{\sigma}, \sigma'\wedge \dot{\sigma} \rangle_2^{1/4}}\right]^2 \nonumber \\
=& \sum_{0 \leq a<b \leq 3}\left[ \frac{\langle \sigma' \wedge \dot{\sigma}, e_a \wedge e_b \rangle_2}{\langle \sigma'\wedge \dot{\sigma}, \sigma'\wedge \dot{\sigma} \rangle_2^{1/4}}\right]^2
=\sum_{0 \leq a<b \leq 3}\rho_\sigma^{ab}|J_{ab}^\sigma|. \label{e.j.4}
\end{align}
This completes the proof.
\end{proof}

Area of a surface is invariant under spatial rotation, but it is not invariant under boost. To construct an unitary representation of the Lorentz group, we need to consider imaginary time-axis. Instead of using $\rho_\sigma^{ab}$ as defined in Equation (\ref{e.ri.1}) for some parametrization $\sigma$ for $S$, we will replace the time component $\sigma_0$ with imaginary time $i \sigma_0$.

\begin{defn}\label{d.r.2}
Let $\sigma: [0,1]^2 \equiv I^2 \rightarrow \bR^4$ be a parametrization of a surface $S \subset \bR^4$.
For $a,b=0,1,2,3$ and $a < b$, define Jacobian matrices,
\begin{align}
\acute{J}_{ab}^\sigma(s,t) =
\left\{
  \begin{array}{ll}
    \ \ \ \ J_{ab}^\sigma(s,t), & \hbox{$a \neq 0$;} \\
    \left(
               \begin{array}{cc}
                 i\sigma_a'(s,t) & i\dot{\sigma}_a(s,t) \\
                 \sigma_b'(s,t) & \dot{\sigma}_b(s,t) \\
               \end{array}
             \right), & \hbox{$a = 0$.}
  \end{array}
\right.
\nonumber
\end{align}

For $a,b, c, d$ all distinct, define $\acute{\rho}_\sigma^{ab}: S \rightarrow \bC$ by
\begin{align*}
\acute{\rho}_\sigma^{ab} := \frac{\det \acute{J}_{ab}^\sigma}{\sqrt{\det\left[\acute{J}_{ab}^{\sigma,T}\acute{J}_{ab}^\sigma + \acute{J}_{cd}^{\sigma,T}\acute{J}_{cd}^\sigma\right]}},
\end{align*}
and
\begin{align*}
\int_S d\acute{\rho} &:= \sum_{0\leq a < b \leq 3}\int_{I^2}\acute{\rho}_\sigma^{ab}(\hat{s})[\det \acute{J}_{ab}^\sigma](\hat{s})\ d\hat{s}, \\
\int_S d|\acute{\rho}| &:= \int_{I^2}\left| \sum_{0\leq a < b \leq 3}\acute{\rho}_\sigma^{ab}(\hat{s})[\det \acute{J}_{ab}^\sigma](\hat{s})\right|\ d\hat{s}.
\end{align*}
\end{defn}

\begin{lem}\label{l.b.1}
Let $S$ be a compact time-like or space-like surface, contained in a plane. Then $\int_S d\acute{\rho}$ and $\int_S d|\acute{\rho}|$ remain invariant under any Lorentz transformation $\Lambda: S \mapsto \hat{S} = \Lambda S$, $\Lambda$ is a $4 \times 4$ Lorentz matrix.
\end{lem}

\begin{proof}
Let $\sigma: I^2 \rightarrow S$ be a parametrization of $S$. Then $\hat{\sigma} = \Lambda \sigma$ is a parametrization of $\hat{S} = \Lambda S$. Since $\vec{x} \cdot \vec{y}$ is invariant under Lorentz transformation, we have
from Equation (\ref{e.j.4}),
\begin{align}
\sum_{0\leq a < b \leq 3}\acute{\rho}_{\hat{\sigma}}^{ab}[\det \acute{J}_{ab}^{\hat{\sigma}}] &= \sqrt{[\hat{\sigma}'\cdot \hat{\sigma}'] [ \dot{\hat{\sigma}}\cdot \dot{\hat{\sigma}} ] - [ \hat{\sigma}'\cdot \dot{\hat{\sigma}} ]^2} \nonumber \\
&= \sqrt{[\sigma' \cdot \sigma'][\dot{\sigma} \cdot \dot{\sigma}] - [\sigma' \cdot \dot{\sigma}]^2} \nonumber \\
&= \sum_{0\leq a < b \leq 3}\acute{\rho}_{\sigma}^{ab}[\det \acute{J}_{ab}^{\sigma}]. \label{e.j.6}
\end{align}
This shows that $\int_{\hat{S}} d\acute{\rho}= \int_{S} d\acute{\rho}$ and $\int_{\hat{S}} d|\acute{\rho}|= \int_{S} d|\acute{\rho}|$.
\end{proof}

\begin{rem}
\begin{enumerate}
  \item When $S$ is a surface in spatial $\bR^3$, we see that $\int_S d|\acute{\rho}| = \int_S d\rho$, which is the area of the surface $S$.
  \item Note that $\int_S d\acute{\rho}$ can be complex valued.
\end{enumerate}
\end{rem}


\begin{thebibliography}{1}

\bibitem{streater}
R.~F.~S. und A~S~Wightman, {\em PCT, Spin Statistics, And All That}.
\newblock New York, Amsterdam: W A Benjamin Inc., 1964.

\bibitem{glimm1981quantum}
J.~Glimm and A.~Jaffe, {\em Quantum physics: a functional integral point of
  view}.
\newblock Springer-Verlag, 1981.

\bibitem{knapp2016representation}
A.~Knapp, {\em Representation Theory of Semisimple Groups: An Overview Based on
  Examples (PMS-36)}.
\newblock Princeton Mathematical Series, Princeton University Press, 2016.

\bibitem{10.2307/1968551}
E.~Wigner, ``On unitary representations of the inhomogeneous lorentz group,''
  {\em Annals of Mathematics}, vol.~40, no.~1, pp.~149--204, 1939.

\bibitem{MR213473}
D.~vong Duc; Nguyen Van~Hieu, ``On the theory of unitary representations of the
  ${\rm sl}(2,\,c)$ group,'' {\em Ann. Inst. H. Poincar\'{e} Sect. A (N.S.)},
  vol.~6, pp.~17--37, 1967.

\bibitem{ALBEVERIO197439}
S.~Albeverio and R.~H{\o}egh-Krohn, ``The wightman axioms and the mass gap for
  strong interactions of exponential type in two-dimensional space-time,'' {\em
  Journal of Functional Analysis}, vol.~16, no.~1, pp.~39 -- 82, 1974.

\bibitem{MR1312606}
R.~W.~R. Darling, {\em Differential forms and connections}.
\newblock Cambridge: Cambridge University Press, 1994.

\bibitem{jaffe2006quantum}
A.~Jaffe and E.~Witten, ``Quantum yang-mills theory,'' {\em The millennium
  prize problems}, no.~1, p.~129, 2006.


\end{thebibliography}
\end{document}